\documentclass[a4paper,11pt,twoside]{article}
\usepackage[a4paper,top=3cm, bottom=3cm, left=3cm, right=3cm]{geometry}
\usepackage[american]{babel}
\usepackage[utf8]{inputenc}       
\usepackage[T1]{fontenc}
\usepackage{fourier}
\usepackage{bbm}
\usepackage{amssymb,amsmath,amsthm}
\usepackage{latexsym}
\usepackage{authblk}
\usepackage[shortlabels]{enumitem}
\usepackage{stmaryrd}
\usepackage{mathtools}
\usepackage{etoolbox}
\usepackage{needspace}
\usepackage{graphicx,color,xcolor}
\usepackage[colorlinks=true,hyperindex=true]{hyperref}
\colorlet{darkblue}{blue!50!black}
\hypersetup{
	colorlinks,%
	citecolor=darkblue,%
	filecolor=red,%
	linkcolor=darkblue,%
	urlcolor=darkblue,%
	pdfnewwindow=true,%
	pdfstartview={FitH}
}
\newcommand{\e}{\mathrm{e}}
\renewcommand{\d}{\mathrm{d}}

\renewcommand{\i}{\mathrm{i}}

\newcommand{\zz}{\mathbb{Z}}
\newcommand{\nn}{\mathbb{N}}
\newcommand{\rr}{\mathbb{R}}

\newcommand{\cH}{\mathcal{H}}
\newcommand{\fA}{\mathfrak{A}}
\newcommand{\fAloc}{{\mathfrak{A}_\mathrm{loc}}}

\newcommand{\cS}{\mathcal{S}}

\newcommand{\cF}{\mathcal{F}}
\newcommand{\cB}{\mathcal{B}}

\newcommand{\cBr}{\mathcal{B}^r}
\newcommand{\cSI}{\mathcal{S}_\mathrm{I}}
\newcommand{\cSeq}{\mathcal{S}_\mathrm{eq}}
\newcommand{\cSwg}{\mathrm{WG}}%{\mathcal{S}_\mathrm{wg}}
\newcommand{\cO}{\mathcal{O}}

\newcommand{\supp}{\mathop{\mathrm{supp}}\nolimits}

\newcommand{\Dom}{\mathop{\mathrm{Dom}}\nolimits}

\newcommand{\wstarlim}{\mathop{\mathrm{w}^\ast-\mathrm{lim}}\limits}
\newcommand{\wstar}[1]{$\mathrm{weak}^\ast$\!-{#1}}
\def\one{\mathbbm{1}}
\newcommand{\lb}{\llbracket}
\newcommand{\rb}{\rrbracket}
\def\be{\begin{equation}}
\def\ee{\end{equation}}

\usepackage{graphicx}
\usepackage{color}

\theoremstyle{plain}
\newtheorem{thm}{Theorem}[section]

\newtheorem{lem}[thm]{Lemma}
\newtheorem{prop}[thm]{Proposition}

\theoremstyle{definition}
\newtheorem{defn}[thm]{Definition}

\AtBeginEnvironment{thm}{\Needspace{5\baselineskip}}% \break if fewer than 5\baselineskip is available on page
\AtBeginEnvironment{prop}{\Needspace{5\baselineskip}}% \break if fewer than 5\baselineskip is available on page
\AtBeginEnvironment{defn}{\Needspace{5\baselineskip}}% \break if fewer than 5\baselineskip is available on page

\usepackage{fancyhdr}
\usepackage{xargs}
\usepackage[textwidth=2.4cm]{todonotes}
\newcommandx{\ca}[2][1=]{\todo[inline,author={ca},
	linecolor=green,backgroundcolor=green!15,bordercolor=green,#1]{#2}}
\newcommandx{\canot}[2][1=]{\todo[
	linecolor=green,backgroundcolor=green!15,bordercolor=green,#1]{#2}}
\newcommandx{\ToDo}[2][1=]{\todo[inline,author={v},
	linecolor=red,backgroundcolor=red!15,bordercolor=red,#1]{#2}}
\newcommandx{\NEW}[2][1=]{\todo[inline,
	linecolor=magenta,backgroundcolor=magenta!15,bordercolor=magenta,#1]{#2}}

\title{A note on adiabatic time evolution and quasi-static processes in translation-invariant quantum systems}
\author[1]{Vojkan Jak\v si\'c}
\author[2]{Claude-Alain Pillet$^\ast$}
\author[3]{Clément Tauber}
\affil[1]{Department of Mathematics and Statistics, McGill University,
	805 Sherbrooke Street West, Montreal, QC, H3A 2K6, Canada}
\affil[2]{Aix Marseille Univ, Université de Toulon, CNRS, CPT, Marseille, France}
\affil[3]{Institut de Recherche Mathématique Avancé, UMR 7501 Université de Strasbourg et CNRS, 7 rue René-Descartes, 67000 Strasbourg, France}
\date{\today}

\begin{document}
\def\today{}	
\maketitle

\centerline{\large \bf Dedicated to the memory of Krzysztof Gaw\k edzki}

\vspace{1cm}
\begin{small}
\noindent{\bf Abstract.}
We study the slowly varying, non-autonomous quantum dynamics of a translation invariant spin or fermion system on the lattice $\zz^d$. This system is  assumed to be initially in thermal equilibrium, and we consider realizations of quasi-static processes in the adiabatic limit. By combining the Gibbs variational principle with the notion of quantum weak Gibbs states introduced in~\cite{Jaksic2023a}, we establish a number of general structural results regarding such realizations. In particular, we show that such a quasi-static process is incompatible with the property of approach to equilibrium studied in this previous work.
\end{small}

%======================
\section{Introduction}
%======================

This paper is a direct continuation of~\cite{Jaksic2023a}. Although its topic and implications are different, we will rely heavily on the technical tools and conceptual framework introduced in this earlier work. In particular, in what follows and without further saying, we will use terminology, notation, and results stated in the introductory section of~\cite{Jaksic2023a}.\footnote{The general references in~\cite{Jaksic2023a} related to the mathematical theory of algebraic quantum statistical mechanics apply to this work as well.}

The set of quantum weak Gibbs states $\cSwg(\Phi)$ for the spin/fermion interaction $\Phi\in\cBr$ is introduced in Section~2.2 of~\cite{Jaksic2023a}. An immediate consequence of our definition of weak Gibbs states is the existence of the specific relative entropy and the validity of the formula
\begin{equation}
s(\nu|\omega)\coloneqq\lim_{\Lambda \uparrow \zz^d}
\frac{S(\nu_\Lambda|\omega_\Lambda)}{|\Lambda|}=-s(\nu)+\nu(E_\Phi)+P(\Phi),
\label{ad-wg}
\end{equation}
for any $\nu\in\cSI(\fA)$ and $\omega\in\cSwg(\Phi)$. The  current technology allowed us to prove the identification $\cSeq(\Phi)=\cSwg(\Phi)$ in dimension $d=1$  for finite range interactions, and for any $d$ and $\Phi \in \cBr$  in the high temperature regime. Since this identification is expected to hold much more generally, we have introduced the notion of {\em regular pair} $(\omega,\Phi)\in\cSI(\fA)\times\cBr$ for which~\eqref{ad-wg} holds for all $\nu\in\cSI(\fA)$. This notion, combined with the Gibbs variational principle, has led to several structural results about the fundamental problem of Approach to Equilibrium---the zeroth law of thermodynamics---in quantum statistical mechanics.

In this work, starting with the same general ingredients, we study whether quasi-static transitions of extended spin or fermion systems can be described by slowly varying, non-auto\-nomous translation invariant interactions. The emerging structural theory in particular yields that a realization of such a quasi-static transition is possible only if the specific entropy is constant along the state trajectory. Combined with the results of~\cite{Jaksic2023a}, this gives that, in the adiabatic limit,  quasi-static transitions  are incompatible with the approach to equilibrium  in the translation invariant setting of algebraic quantum statistical mechanics.

The paper is organized as follows. In Sections~\ref{sec-old-ad} and~\ref{sec-old-ad-1}  we briefly review the basic adiabatic theorems of quantum mechanics and quantum statistical mechanics. For a more complete review, historical perspective, and additional references, we refer the reader to~\cite{Benoist2023}.\footnote{In particular, we have not attempted here to formulate these results in a technically optimal setting.} We introduce our translation invariant framework in Section~\ref{sec-setting}. Our main results are stated and proved in Section~\ref{sec-main-re}.  They are compared with the main result of~\cite{Jaksic2023a} in Section~\ref{sec-no-go}. Similar to~\cite{Jaksic2023a}, most of our proofs are technically simple. The proofs of two technically more involved results, Propositions~\ref{prop-balance} and~\ref{prop:alpha}, are postponed to Section~\ref{sec-ln}. For the reader's convenience, Appendix~\ref{app:notations} collects the notations inherited from~\cite{Jaksic2023a}. 

\paragraph*{Acknowledgments} This work was supported by the French {\em Agence Nationale de la Recherche}, grant NONSTOPS (ANR-17-CE40-0006-01, ANR-17-CE40-0006-02, ANR-17-CE40-0006-03) and the CY Initiative of Excellence through the grant Investissements d’Avenir ANR-16-IDEX-0008. It was partly developed during VJ's stay at the CY Advanced Studies, whose support is gratefully acknowledged. VJ also acknowledges the support of NSERC. The authors wish to thank Martin Fraas for useful discussions.

%=================================================
\section{Adiabatic theorems in quantum mechanics}
\label{sec-old-ad}
%=================================================

{\bf The adiabatic theorem with a gap.} The formulation and the proof of the adiabatic theorem in quantum mechanics go back to the seminal work of Born and Fock~\cite{Born1928}. The first modern proof, in functional analytic framework, is due to  Kato~\cite{Kato1950}. The method of the proof introduced there played an important role in all the future developments of the subject. For refinements of Kato's result  under the gap assumption made in~\cite{Kato1950}, see~\cite{Avron1987,Joye1991a,Nenciu1993}.

Let $H$ be a (possibly unbounded) self-adjoint operator on the Hilbert space $\cH$ and 
\begin{equation}
V:[0,1]\to\cB(\cH)
\label{ear-ear}
\end{equation}
a continuous map taking its values in the bounded self-adjoint operators on $\cH$. For $T>0$, we consider the non-autonomous Schrödinger equation
$$
\i\partial_t\psi(t)=(H+V(t/T))\psi(t),\qquad\psi(s)=f,\qquad s,t\in[0,T].
$$
For any $f\in\Dom(H)$ this equation has a unique solution $[0,T]\ni t\mapsto\psi(t)\in\Dom(H)$, which can be written as $\psi(t)=U_T^{s/T\to t/T}f$, where $[0,1]\times[0,1]\ni(\sigma,\tau)\mapsto U_T^{\sigma\to\tau}\in\cB(\cH)$ is a strongly continuous map taking its values in the unitary operators on $\cH$ and satisfying $U_T^{\sigma\to\sigma}=I$, $U_T^{\sigma\to\tau}\Dom(H)\subset\Dom(H)$, and $U_T^{\tau\to\upsilon}U_T^{\sigma\to\tau}=U_T^{\sigma\to\upsilon}$ for all $\sigma,\tau,\upsilon\in[0,1]$. $U_T$ is the propagator associated to the time-dependent Hamiltonian $[0,1]\ni\tau\mapsto T(H+V(\tau))$.

For each $\tau\in[0,1]$, let $E(\tau)$ be an eigenvalue of finite multiplicity of $H(\tau)=H+V(\tau)$ which is  uniformly separated from the remaining part of its spectrum, {\sl i.e.,} 
\begin{equation}
\inf_{\tau\in[0,1]}{\rm dist}(E(\tau),{\rm sp}(H(\tau))\setminus\{E(\tau)\})>0,
\label{kato-gap}
\end{equation}
and denote by $P(\tau)$ the orthogonal projection onto the associated eigenspace. 

\begin{thm}{\bf\!\cite{Kato1950}}\label{kato-thm}
Suppose, in addition to~\eqref{kato-gap}, that the map~\eqref{ear-ear} is $C^2$. Then,  as $T\uparrow \infty$,
$$
\sup_{\tau\in[0,1]}\|(I-P(\tau))U_T^{0\to\tau}P(0)\|=O(T^{-1}).
$$
\end{thm}

\noindent{\bf The adiabatic theorem without a gap.} This fundamental refinement of Theorem~\ref{kato-thm} goes back to Avron and Elgart~\cite{Avron1999}, who dispense with the gap assumption~\eqref{kato-gap}; for an important technical comment on their result see~\cite{Teufel2001}. The setting is the same as in Theorem~\ref{kato-thm}, except that~\eqref{kato-gap} is replaced by the following  assumption:
\begin{quote}{\bf Assumption (AD)} There exists a $C^2$ map $[0,1]\ni\tau\mapsto P(\tau)\in\cB(\cH)$ such that,
for Lebesgue a.e.\;$\tau\in[0,1]$, $P(\tau)$ is the orthogonal projection onto the eigenspace of a finite multiplicity eigenvalue of $H(\tau)$. 
\end{quote}
\begin{thm}{\bf~\cite{Avron1999,Teufel2001}}\label{ae-thm}
Suppose that {\bf (AD)} holds. Then 
$$
\lim_{T\rightarrow\infty}\,\sup_{\tau\in[0,1]}\|(I-P(\tau))U_T^{0\to\tau}P(0)\|=0.
$$
\end{thm}

%=============================================================
\section{Adiabatic theorems in quantum statistical mechanics}
\label{sec-old-ad-1}
%=============================================================

{\bf The isothermal adiabatic theorem for local perturbations.} Let $(\cO,\alpha)$ be a $C^\ast$-dynamical system. Denote by $\delta$ its generator, $\alpha^t=\e^{t\delta}$, and by $\omega$ an $\alpha$-KMS state at inverse temperature $\beta>0$. To a self-adjoint $V\in\cO$ we associate the perturbed dynamics $\alpha_V^t=\e^{t\delta_V}$, where $\delta_V=\delta+\i[V,\,\cdot\,]$, and the perturbed $(\alpha_V,\beta)$-KMS state $\omega_V$.

Let
\begin{equation}
[0,1]\ni\tau\mapsto V(\tau)\in\cO
\label{map-ear}
\end{equation}
be a continuous map such that $V(\tau)$ is self-adjoint for all $\tau$. For $T>0$, the unique solution of the Cauchy problem for the non-autonomous Heisenberg equation
$$
\partial_t\gamma^t(A)=\gamma^t\circ\delta_{V(t/T)}(A),\qquad \gamma^s(A)=A\in\Dom(\delta),\qquad s,t\in[0,T],
$$
is given by $\gamma^t=\alpha_T^{s/T\to t/T}$, where $[0,1]\times[0,1]\ni(\sigma,\tau)\mapsto\alpha_T^{\sigma\to\tau}$ is a strongly continuous two-parameter family of $\ast$-automorphisms of $\cO$ satisfying the relation $\alpha_T^{\sigma\to\tau}\circ\alpha_T^{\tau\to\upsilon}=\alpha_T^{\sigma\to\upsilon}$ for any $\sigma,\tau,\upsilon\in[0,1]$. For $0\le\sigma\le\tau\le 1$, it has the norm-convergent expansion\footnote{It is understood that the zeroth term in this expansion is $\alpha^{(\tau-\sigma)T}(A)$. A similar expression holds for $0\le\tau\le\sigma\le1$.}
$$
\alpha_T^{\sigma\to\tau}(A)=
\sum_{n=0}^\infty\hskip6pt T^n\hskip-16pt\int\limits_{\sigma\leq\sigma_1\leq\cdots\leq\sigma_n\leq\tau}\hskip-14pt
\i[\alpha^{(\sigma_1-\sigma)T}(V(\sigma_1)),\i[\cdots,\i[\alpha^{(\sigma_n-\sigma)T}(V(\sigma_n)),\alpha^{(\tau-\sigma)T}(A)]\cdots]]\d\sigma_1\cdots\d\sigma_n.
$$

Recall the property of Return to Equilibrium reviewed in~\cite[Section~1.1]{Jaksic2023a}, where the reader can also find references to this well-studied topic. 

\begin{thm}\label{adl-thm} Suppose that the map~\eqref{map-ear} is $C^2$. If, for Lebesgue a.e.\;$\tau\in [0,1]$, the  quantum dynamical system $(\cO,\alpha_{V(\tau)},\omega_{V(\tau)})$ has the property of return to equilibrium, then 
\be
\lim_{T\to\infty}\,\sup_{\tau\in[0,1]}\|\omega_{V(0)}\circ\alpha_T^{0\to\tau}-\omega_{V(\tau)}\|=0.
\label{ln-fox}
\ee
\end{thm}
This  result goes back to~\cite{AbouSalem2007a,Jaksic2020}; see also~\cite{Jaksic2014b}. The proof is a simple combination of the Avron--Elgart gapless adiabatic theorem~\ref{ae-thm} and  Araki's perturbation theory of the modular structure. 

For an in depth discussion of Theorem~\ref{adl-thm} and its thermodynamical implications we refer the reader to~\cite{Benoist2023}. Here we recall three results that will be of relevance in what follows; see also~\cite[Section 5]{Jaksic2014b} for a related discussion.

The first is the following entropy balance equation: 
\begin{prop}\label{prop-balance}
Suppose that the map~\eqref{map-ear} is $C^1$ and denote by $\dot V$ its derivative. Then\footnote{ $S(\,\cdot\,|\,\cdot\,)$ is the relative entropy functional, with the sign and ordering convention of~\cite[Section~1.1]{Jaksic2023a}.} 
\begin{equation*}
S\left(\omega_{V(0)}\circ\alpha_T^{0\to\tau}\bigg|\omega_{V(\tau)}\right)
=\beta\int_0^\tau\left(\omega_{V(0)}\circ\alpha_T^{0\to\sigma}-\omega_{V(\sigma)}\right)(\dot V(\sigma))\d\sigma
\end{equation*}
holds for all $\tau\in [0,1]$.
\end{prop}
For the reader's convenience, a proof of this proposition is given in Section~\ref{sec-proof-of-prop:balance}. 

The entropy balance equation yields the estimate
\[
\sup_{\tau\in[0,1]}S\left(\omega_{V(0)}\circ\alpha_T^{0\to\tau}\bigg|\omega_{V(\tau)}\right)
\le\beta\sup_{\tau\in[0,1]}\|\omega_{V(0)}\circ\alpha_T^{0\to\tau}-\omega_{V(\tau)}\|
\,\sup_{\tau\in[0,1]}\|\dot V(\tau)\|,
\]
and the validity of~\eqref{ln-fox} implies  that the adiabatic theorem for local perturbations also holds in the entropic sense:
\begin{equation}
\lim_{T\to\infty}\,\sup_{\tau\in[0,1]}
S\left(\omega_{V(0)}\circ\alpha_T^{0\to\tau}\bigg|\omega_{V(\tau)}\right)=0.
\label{ad-relent}
\end{equation}
On the other hand, the Pinsker--Csiszàr inequality~\cite[Proposition~5.23]{Ohya1993}, $\|\omega-\nu\|^2\le2 S(\omega|\nu)$, gives 
that~\eqref{ad-relent} implies~\eqref{ln-fox}.  In summary: 
\begin{thm}\label{prop:so-so-t-0}
Suppose that the map~\eqref{map-ear} is $C^1$. Then, the following statements are equivalent:
\begin{enumerate}[label=(\roman*)]
\item 
\[
\lim_{T\to\infty}\,\sup_{\tau\in[0,1]}
S\left(\omega_{V(0)}\circ\alpha_T^{0\to\tau}\bigg|\omega_{V(\tau)}\right)=0.
\]
\item 
\[\lim_{T\to\infty}\,\sup_{\tau\in[0,1]}\|\omega_{V(0)}\circ\alpha_T^{0\to\tau}-\omega_{V(\tau)}\|=0.\]
\end{enumerate}
\end{thm}

Starting with Proposition~\ref{prop-balance}, a similar argument yields the third result  which is also of independent interest. 
\begin{thm}\label{prop:so-so-t}
Suppose that the map~\eqref{map-ear} is $C^1$. Then, the following statements are equivalent:
\begin{enumerate}[label=(\roman*)]
\item For Lebesgue a.e.\;$\tau\in[0,1]$, 
\[
\lim_{T\to\infty}\omega_{V(0)}\circ\alpha_T^{0\to\tau}(\dot V(\tau))=\omega_{V(\tau)}(\dot V(\tau)).
\]
\item For all $\tau\in[0,1]$,
$$
\lim_{T\to\infty}S\left(\omega_{V(0)}\circ\alpha_T^{0\to\tau}\bigg|\omega_{V(\tau)}\right)=0.
$$
\item For all $\tau\in[0,1]$,
$$
\lim_{T\to\infty}\|\omega_{V(0)}\circ\alpha_T^{0\to\tau}-\omega_{V(\tau)}\|=0.
$$
\end{enumerate}
\end{thm}
In Section~\ref{sec-main-re} we extend Theorem~\ref{prop:so-so-t} to translation invariant many body systems. 

\bigskip
\noindent{\bf The adiabatic theorem for gapped ground states of quantum spin systems.} A very general extension of Theorem~\ref{kato-thm} to many body quantum spin systems was established in~\cite{Bachmann2017,Bachmann2018}.  Since this result plays no role in our work, we will not discuss it further. The low temperature adiabatic theory of 
lattice fermion systems was studied in the recent work~\cite{Greenblatt2022}. The relation between  our work and  \cite{Greenblatt2022} remains to be studied further.

%===========================
\section{Structural theory}
\label{sec-structural}
%===========================

%===================
\subsection{Setup} 
\label{sec-setting}
%===================

As in~\cite{Jaksic2023a}, we will work with quantum spin systems on the lattice $\zz^d$ and set the inverse temperature to $\beta=1$. All the results and proofs directly extend to the fermionic case. The notation and terminology are the same as in~\cite{Jaksic2023a}; see Appendix~\ref{app:notations} below for a brief summary. 

Given an interaction $\Phi\in\cBr$, the associated local Hamiltonian in a finite cube $\Lambda$ is
\[
H_\Lambda(\Phi)=\sum_{X\subset\Lambda}\Phi(X).
\]
The $C^\ast$-dynamics it generates on $\fA$ is
$$
\alpha_{\Phi,\Lambda}^t(A)=\e^{\i tH_\Lambda(\Phi)}A\e^{-\i tH_\Lambda(\Phi)},
$$
and its thermodynamic limit
$$
\alpha_{\Phi}^t(A)=\lim_{\Lambda\uparrow\zz^d}\alpha_{\Phi,\Lambda}^t(A),
$$
yields a $C^\ast$-dynamics commuting with the natural action of the translation group $\zz^d\ni x\mapsto\varphi^x$.

We consider time-dependent translation invariant interactions described by continuous maps
\[
\Phi:[0,1]\ni\tau\mapsto\Phi_\tau\in\cBr.
\]
In analogy with Section~\ref{sec-old-ad-1}, the local time-dependent Hamiltonians  $H_\Lambda(\Phi_\tau)$ generate a non-autonomous $C^\ast$-dynamics $\alpha_{\Phi,\Lambda}^{\sigma\to\tau}$ on $\fA$, uniquely characterized by the Cauchy problem for the Heisenberg equation of motion
\[
\partial_\tau\alpha_{\Phi,\Lambda}^{\sigma\to\tau}(A)
=\alpha_{\Phi,\Lambda}^{\sigma\to\tau}(\i[H_\Lambda(\Phi_\tau),A]),
\qquad  
\alpha_{\Phi,\Lambda}^{\sigma\to\sigma}(A)=A\in\fA,\qquad\sigma,\tau\in[0,1].
\]
The following proposition summarizes the basic properties of its thermodynamic limit needed in this paper. The reader
should consult~\cite{Nachtergaele2019} for a more detailed and more general discussion of non-autonomous $C^\ast$-dynamics. 
\begin{prop}\label{prop:alpha}
Let $\Phi:[0,1]\to\cBr$ be a continuous, time-dependent, translation invariant interaction.
\begin{enumerate}[label=(\arabic*)]
\item For all $\sigma,\tau\in[0,1]$ and $A\in\fA$, the limit 
\be
\alpha_{\Phi}^{\sigma\to\tau}(A)=\lim_{\Lambda\uparrow\zz^d}\alpha_{\Phi,\Lambda}^{\sigma\to\tau}(A)
\label{equ:alphaTDlimit}
\ee
exists and defines a two-parameter family of $\ast$-auto\-morphisms of\/ $\fA$ such that
\begin{equation*}
\alpha_{\Phi}^{\sigma\to\tau}\circ\alpha_{\Phi}^{\tau\to\upsilon}=\alpha_{\Phi}^{\sigma\to\upsilon},
\qquad  
\alpha_{\Phi}^{\sigma\to\sigma}=\mathrm{Id},
\end{equation*}
for all $\sigma,\tau,\upsilon\in[0,1]$.
\item $\alpha_\Phi$ commutes with the group action of\/ $\zz^d$, {\sl i.e.,}
$$
\alpha_\Phi^{\sigma\to\tau}\circ\varphi^x=\varphi^x\circ\alpha_\Phi^{\sigma\to\tau},
$$
for all $\sigma,\tau\in[0,1]$ and $x\in\zz^d$.
\item The map $[0,1]\ni\sigma,\tau\mapsto\alpha_\Phi^{\sigma\to\tau}\in\mathrm{Aut}(\fA)$ is strongly continuous. Moreover, for all $\sigma,\tau\in{]}0,1{[}$ and any $A\in\fAloc$,
$$
\partial_\tau\alpha_{\Phi}^{\sigma\to\tau}(A)=\alpha_{\Phi}^{\sigma\to\tau}\circ\delta_{\Phi_\tau}(A),
$$
where $\delta_{\Phi_\tau}$ denotes the map defined by
$$
\delta_{\Phi_\tau}(A)=\sum_{X\cap\supp(A)\neq\emptyset}\i[\Phi_\tau(X),A].
$$
\item For $\sigma,\tau\in[0,1]$ and integer $N\ge1$, set $\xi_N=(\tau-\sigma)/N$ and $\upsilon_k=\sigma+k\xi_N$. Then, for $A\in\fA$, one has
$$
\alpha_\Phi^{\sigma\to\tau}(A)=\lim_{N\to\infty}
\alpha_{\Phi_{\upsilon_0}}^{\xi_N}\circ\cdots\circ\alpha_{\Phi_{\upsilon_{N-1}}}^{\xi_N}(A).
$$
\end{enumerate}
\end{prop}

A proof is given in Section~\ref{sec-proof-of-prop:alpha}. 
The family $\{\alpha_{\Phi}^{\sigma\to\tau}\}_{\sigma,\tau\in[0,1]}$ defines the non-autonomous $C^\ast$-dynamics generated by the time-dependent interaction $\Phi$. 

The basic result on  the constancy of  specific entropy along the state trajectory, \cite[Theorem~5]{Lanford1968}, extends to time-dependent interactions. 

\begin{prop}\label{LR-con-1}
For any $\nu\in\cSI(\fA)$ and $\sigma,\tau\in[0,1]$,
\[
s(\nu)=s(\nu\circ\alpha_{\Phi}^{\sigma\to\tau}).
\]
\end{prop}
\begin{proof} Invoking reversibility, it suffices to prove that, for all $\sigma,\tau\in[0,1]$, one has 
\begin{equation}
s(\nu)\leq s(\nu\circ\alpha_{\Phi}^{\sigma\to\tau}).
\label{pa-tr}
\end{equation}
Part~(4) of Proposition~\ref{prop:alpha} gives
$$
s(\nu\circ\alpha_{\Phi}^{\sigma\to\tau})
=s\left(\wstarlim_{N\to\infty}
\nu\circ\alpha_{\Phi_{\upsilon_0}}^{\xi_N}\circ\cdots\circ\alpha_{\Phi_{\upsilon_{N-1}}}^{\xi_N}\right),
$$
and by~\cite[Theorem~5]{Lanford1968}, we have
$$
s(\nu\circ\alpha_{\Phi_{\upsilon_0}}^{\xi_N}\circ\cdots\circ\alpha_{\Phi_{\upsilon_{N-1}}}^{\xi_N})=s(\nu).
$$
The upper-semicontinuity of specific entropy yields~\eqref{pa-tr}. 
\end{proof}

%=========================
\subsection{Main results}
\label{sec-main-re}
%=========================

In this section we associate to a time-dependent interaction $\Psi$ a family of instantaneous equilibrium states. More precisely, we will denote by $\lb\nu,\Psi\rb$ a continuous map
\[
[0,1]\ni\tau\mapsto(\nu_\tau,\Psi_\tau)\in\cSI(\fA)\times\cBr
\]
such that $\nu_\tau\in\cSeq(\Psi_\tau)$ for all $\tau\in[0,1]$. We will consider the non-autonomous time evolution on $\fA$ 
defined by the Cauchy problem
$$
\partial_t\gamma^t(A)=\gamma^t\circ\delta_{\Psi_{t/T}}(A),\qquad \gamma^s(A)=A\in\fAloc,\qquad s,t\in[0,T],
$$
in the adiabatic limit $T\to\infty$. A rescaling of the time variables $s$ and $t$ gives
$$
\gamma^{t}=\alpha_{T\Psi}^{s/T\to t/T},
$$ 
where $\{\alpha_{T\Psi}^{\sigma\to\tau}\}_{\sigma,\tau\in[0,1]}$ is the two-parameter family of $\ast$-automorphisms generated by the re\-scaled interaction $T\Psi$, as described in Proposition~\ref{prop:alpha}.

\begin{defn}\label{def:adia}
\hspace{0em}%for the bug in hyperref
\begin{enumerate}[(1)]
\item A regular pair $(\nu,\Psi)\in\cSI(\fA)\times\cBr$ is {\sl uniquely regular\/} whenever $\cSeq(\Psi)=\{\nu\}$. A path $\lb\nu,\Psi\rb$ is uniquely regular whenever $(\nu_\tau,\Psi_\tau)$ is uniquely regular for all $\tau\in[0,1]$.

\item We say that $\lb\nu,\Psi\rb$ satisfies the \emph{adiabatic theorem} if
\[
\wstarlim_{T\to\infty}\nu_0\circ\alpha_{T\Psi}^{0\to1}=\nu_1,
\]
and the \emph{path adiabatic theorem} if, for all $\tau\in[0,1]$, 
\[
\wstarlim_{T\to\infty}\nu_0\circ\alpha_{T\Psi}^{0\to\tau}=\nu_\tau.
\]
\item  We say that $\lb\nu,\Psi\rb$ satisfies the \emph{entropic adiabatic theorem} if
\[
\lim_{T\to\infty}s(\nu_0\circ\alpha_{T\Psi}^{0\to1}|\nu_1)=0,
\]
and the \emph{entropic path adiabatic theorem} if, for all $\tau\in[0,1]$, 
\[
\lim_{T\to\infty}s(\nu_0\circ\alpha_{T\Psi}^{0\to\tau}|\nu_\tau)=0.
\]
\end{enumerate}
\end{defn}

In the fermionic case, gauge-invariance implies that the adiabatic theorem can hold only if $\nu_0(E_N)=\nu_1(E_N)$. Similarly, the path adiabatic theorem can hold only if $\nu_0(E_N)=\nu_\tau(E_N)$ for all $\tau\in[0,1]$. In what follows we will discuss only the quantum spin case, and leave the elementary  reformulations needed to accommodate the fermionic setting to the interested reader.

\begin{prop}\label{prop:nu0} Suppose that  $(\nu_1,\Psi_1)$ is a uniquely regular pair. Then the following statements are equivalent:
\begin{enumerate}[label=(\roman*)]
\itemsep0em
\item $s(\nu_0)=s(\nu_1)$ and the adiabatic theorem holds for $\lb\nu,\Psi\rb$. 
\item $s(\nu_0)=s(\nu_1)$ and $\displaystyle\lim_{T\to\infty}\nu_0\circ\alpha_{T\Psi}^{0\to1}(E_{\Psi_1})=\nu_1(E_{\Psi_1})$.
\item The entropic adiabatic theorem holds for $\lb\nu,\Psi\rb$.
\end{enumerate}
\end{prop}\begin{proof}

Note that $\nu_0\circ\alpha_{T\Psi}^{0\to1}\in\cSI(\fA)$ by Part~(2) of Proposition~\ref{prop:alpha}.
	
(i)$\Rightarrow$(ii). Follows directly from Definition~\ref{def:adia}.
 
\medskip\noindent(ii)$\Rightarrow$(iii). The regularity of $(\nu_1,\Psi_1)$ gives
\be
s(\nu_0\circ\alpha_{T\Psi}^{0\to1}|\nu_1)=-s(\nu_0\circ\alpha_{T\Psi}^{0\to1})+\nu_0\circ\alpha_{T\Psi}^{0\to1}(E_{\Psi_1})+P(\Psi_1).
\label{zima-1}
\ee
The assumption  $s(\nu_0)=s(\nu_1)$ and  Proposition~\ref{LR-con-1} yield that  $s(\nu_0\circ\alpha_{T\Psi}^{0\to1})=s(\nu_1)$, and so
$$
\lim_{T\to\infty}s(\nu_0\circ\alpha_{T\Psi}^{0\to1}|\nu_1)=-s(\nu_1)+\nu_1(E_{\Psi_1})+P(\Psi_1)=0,
$$
where we used that $\nu_1\in\cSeq(\Psi_1)$.

\medskip\noindent(iii)$\Rightarrow$(i). Let $\nu_+$ be a \wstar{limit} point of the net $(\nu_0\circ\alpha_{T\Psi}^{0\to1})_{T>0}$ as $T\to\infty$. The regularity of $(\nu_1,\Psi_1)$ implies that the map $\cSI(\fA)\ni\omega\mapsto s(\omega|\nu_1)$ is lower-semicontinuous, and so
\[
0\le s(\nu_+|\nu_1)\le\liminf_{T\to\infty}s(\nu_0\circ\alpha_{T\Psi}^{0\to1}|\nu_1)=0.
\]
Since $\nu_+\in\cSI(\fA)$, we deduce from
$$
0=s(\nu_+|\nu_1)=-s(\nu_+)+\nu_+(E_{\Psi_1})+P(\Psi_1)
$$
that $\nu_+\in\cSeq(\Psi_1)=\{\nu_1\}$. It follows that $\nu_0\circ\alpha_{T\Psi}^{0\to1}\rightharpoonup\nu_1$ as $T\to\infty$.  

To prove that 
$s(\nu_0)=s(\nu_1)$, note that the formula~\eqref{zima-1} and Proposition~\ref{LR-con-1} yield 
\[s(\nu_0\circ\alpha_{T\Psi}^{0\to1}|\nu_1)=-s(\nu_0)+\nu_0\circ\alpha_{T\Psi}^{0\to1}(E_{\Psi_1})+P(\Psi_1).\]
Hence, 
\[0=\lim_{T\rightarrow \infty} s(\nu_0\circ\alpha_{T\Psi}^{0\to1}|\nu_1)=-s(\nu_0)+ \nu_1(E_{\Psi_1})+P(\Psi_1)=s(\nu_1)-s(\nu_0),
\]
where we used that $\nu_1\in\cSeq(\Psi_1)$.
\end{proof}

An immediate consequence of the last result is:
\begin{thm}\label{si-2}
Suppose that the path $\lb\nu,\Psi\rb$ is uniquely regular. Then  the following statements are equivalent:
\begin{enumerate}[label=(\roman*)]
\itemsep0em
\item $s(\nu_0)=s(\nu_\tau)$ for all $\tau\in[0,1]$ and the path adiabatic theorem holds for $\lb\nu,\Psi\rb$. 
\item $s(\nu_0)=s(\nu_\tau)$  and $\displaystyle\lim_{T\to\infty}\nu_0\circ\alpha_{T\Psi}^{0\to\tau}(E_{\Psi_\tau})=\nu_\tau(E_{\Psi_\tau})$ for all $\tau\in[0,1]$.
\item The entropic path adiabatic theorem holds for $\lb\nu,\Psi\rb$.
\end{enumerate}
\end{thm}

Our next result is: 

\begin{thm}\label{prop-sun}
Suppose that the path $\lb\nu,\Psi\rb$ is uniquely regular. If $\Psi\in C^1([0,1],\cBr)$, then the following statements are equivalent:
\begin{enumerate}[label=(\roman*)]
\itemsep0em
\item For Lebesgue a.e.\;$\tau\in[0,1]$,
$$
\lim_{T\to\infty}\nu_0\circ\alpha_{T\Psi}^{0\to\tau}(\partial_\tau E_{\Psi_\tau})
=\nu_\tau(\partial_\tau E_{\Psi_\tau}).
$$
\item The entropic path adiabatic theorem holds for $\lb\nu,\Psi\rb$.
\end{enumerate}	
\end{thm}
\begin{proof} (ii)$\Rightarrow$(i) follows from Theorem~\ref{si-2}.

\noindent(i)$\Rightarrow$(ii). The map $[0,1]\ni\tau\mapsto E_{\Psi_\tau}$ is continuously differentiable by assumption. Moreover, it follows from Part~(3) of Proposition~\ref{prop:alpha} that for $A\in\fAloc$,
$$
\alpha_{T\Psi}^{\sigma\to\tau}(A)
=A+T\int_\sigma^\tau\alpha_{T\Psi}^{\sigma\to\upsilon}\circ\delta_{\Psi_\upsilon}(A)\d\upsilon.
$$
By the proof of Part~(2) of~\cite[Theorem~2.8]{Jaksic2023a}, for any $\tau,\upsilon\in[0,1]$ one has $E_{\Psi_\tau}\in\Dom(\delta_{\Psi_\upsilon})$ and 
\be
\omega\circ\delta_{\Psi_\tau}(E_{\Psi_\tau})=0\ \text{for all }\omega\in\cSI(\fA).
\label{equ:magic}
\ee 
Moreover, there exist a sequence $(E_{\Psi_\tau,n})_{n\in\nn}$ in $\fAloc$ such that\footnote{The norm $\|\Psi\|_r$ of a time-dependent interaction $\Psi$ is defined in~\eqref{eq:Jeanne}.} $\|\delta_{\Psi_\upsilon}(E_{\Psi_{\tau,n}})\|\le2\|\Psi\|_r^2$, and
$$
\lim_{n\to\infty}E_{\Psi_\tau,n}=E_{\Psi_\tau},\qquad
\lim_{n\to\infty}\delta_{\Psi_\upsilon}(E_{\Psi_\tau,n})=\delta_{\Psi_\upsilon}(E_{\Psi_\tau}).
$$
Hence, the dominated convergence theorem yields
$$
\alpha_{T\Psi}^{\sigma\to\tau}(E_{\Psi_\tau})
=E_{\Psi_\tau}+T\int_\sigma^\tau\alpha_{T\Psi}^{\sigma\to\upsilon}\circ\delta_{\Psi_\upsilon}(E_{\Psi_\tau})\d\upsilon,
$$
from which we deduce
$$
\partial_\tau\nu_0\circ\alpha_{T\Psi}^{0\to\tau}(E_{\Psi_\tau})
=T\nu_0\circ\alpha_{T\Psi}^{0\to\tau}\circ\delta_{\Psi_\tau}(E_{\Psi_\tau})
+\nu_0\circ\alpha_{T\Psi}^{0\to\tau}(\partial_\tau E_{\Psi_\tau}).
$$
Since $\nu_0\circ\alpha_{T\Psi}^{0\to\tau}\in\cSI(\fA)$ by Part~(2) of Proposition~\ref{prop:alpha}, it follows from~\eqref{equ:magic} that
$$
\nu_0\circ\alpha_{T\Psi}^{0\to\tau}(E_{\Psi_\tau})-\nu_0(E_{\Psi_0})
=\int_0^\tau\nu_0\circ\alpha_{T\Psi}^{0\to\sigma}(\partial_\sigma E_{\Psi_\sigma})\d\sigma.
$$
Taking the limit $T\to\infty$, it follows from our hypotheses and the dominated convergence theorem that
\be
\lim_{T\to\infty}\nu_0\circ\alpha_{\Psi,T}^{0\to\tau}(E_{\Psi_\tau})-\nu_0(E_{\Psi_0})
=\int_0^\tau\nu_\sigma(\partial_\sigma E_{\Psi_\sigma})\d\sigma.
\label{equ:rainyday}
\ee
The uniqueness of the equilibrium state for $\Psi_\tau$ implies that the pressure $\cB^r\ni\Phi\mapsto P(\Phi)$ is differentiable at $\Psi_\tau$, with
$$
P(\Psi_\tau+\Phi)-P(\Psi_\tau)+\nu_\tau(E_\Phi)=o(\|\Phi\|_r)
$$
for $\Phi\in\cBr$ \cite[Theorem~3]{Lanford1968}. Invoking again the differentiability of $\tau\mapsto E_{\Psi_\tau}$ yields 
\be
\partial_\tau P(\Psi_\tau)=-\nu_\tau(\partial_\tau E_{\Psi_\tau}).
\label{re-ln}
\ee
Combining  Relations~\eqref{re-ln} and~\eqref{equ:rainyday} we derive
$$
\lim_{T\to\infty}\nu_0\circ\alpha_{\Psi,T}^{0\to\tau}(E_{\Psi_\tau})+P(\Psi_\tau)=\nu_0(E_{\Psi_0})+P(\Psi_0).
$$
Now, since $(\nu_\tau,\Psi_\tau)$ is a regular pair, we have
$$
s(\nu_0\circ\alpha_{\Psi,T}^{0\to\tau}|\nu_\tau)
=-s(\nu_0)+\nu_0\circ\alpha_{\Psi,T}^{0\to\tau}(E_{\Psi_\tau})+P(\Psi_\tau),
$$
and so
$$
\lim_{T\to\infty}s(\nu_0\circ\alpha_{\Psi,T}^{0\to\tau}|\nu_\tau)=-s(\nu_0)+\nu_0(E_{\Psi_0})+P(\Psi_0)=0.
$$
\end{proof}

The following example illustrates the last result. Let $\Phi_0,\Phi_1\in\cBr$, and let $\lambda\in C^1([0,1])$ be such that $\lambda(0)=0$, $\lambda(1)=1$. Set
\[
\Psi_\tau=\Phi_0+\lambda(\tau)(\Phi_1-\Phi_0),
\]
and assume that the path $\lb\nu,\Psi\rb$ is uniquely regular. Then the entropic  path adiabatic theorem holds for $\lb\nu,\Psi\rb$ iff, for all $\tau\in[0,1]$ with $\partial_\tau \lambda(\tau)\not=0$, 
\[
\lim_{T\rightarrow\infty}\nu_0\circ\alpha_{T\Psi}^{0\to\tau}(E_{\Phi_1-\Phi_0})=\nu_\tau(E_{\Phi_1-\Phi_0}).
\]

We finish this section  with  the translation invariant version Theorem~\ref{prop:so-so-t}. 
This result is direct a consequence  Theorems~\ref{si-2} and~\ref{prop-sun} and  we state it for completeness reason.

\begin{thm}\label{prop-sun-1}
Suppose that the path $\lb\nu,\Psi\rb$ is uniquely regular. If $\Psi\in C^1([0,1],\cBr)$, then the following statements are equivalent:
\begin{enumerate}[label=(\roman*)]
\itemsep0em
\item For Lebesgue a.e.\;$\tau\in[0,1]$,
$$
\lim_{T\to\infty}\nu_0\circ\alpha_{T\Psi}^{0\to\tau}(\partial_\tau E_{\Psi_\tau})
=\nu_\tau(\partial_\tau E_{\Psi_\tau}).
$$
\item The entropic path adiabatic theorem holds for $\lb\nu,\Psi\rb$.
\item The path adiabatic theorem holds for $\lb\nu,\Psi\rb$.
\end{enumerate}	
Moreover, any of the above statements implies
\begin{enumerate}[label=(\roman*)]
\itemsep0em
\setcounter{enumi}{3}
\item For all $\tau\in[0,1]$,
$$
s(\nu_\tau)=s(\nu_0).
$$
\end{enumerate}
\end{thm}

%=========================================================
\subsection{Adiabaticity and the approach to equilibrium}
\label{sec-no-go}
%=========================================================

Theorem~\ref{adl-thm} points to a close relation between the adiabatic theorem for local perturbations of KMS states and return to equilibrium. In view of the formal analogy between Proposition~\ref{prop:so-so-t} and Theorem~\ref{prop-sun-1}, one may expect a similar relation between the adiabatic theorem for translation invariant interactions and approach to equilibrium. Our last result shows that this is not the case.

We first recall the framework of approach to equilibrium developed in~\cite{Jaksic2023a}. Let $\omega\in\cSI(\fA)$ and $\Phi\in\cBr$. Set
\[
\bar\omega_T=\frac{1}{T}\int_0^T\omega\circ\tau_\Phi^t \d t.
\]
We denote by $\cS_+(\omega, \Phi)$ the set of \wstar{limit} points of the net $(\bar\omega_T)_{T>0}$ as $T\to\infty$. Approach to equilibrium holds in this setting whenever $\cS_+(\omega, \Phi)\cap\cSeq(\Phi)\neq\emptyset$.
\begin{thm} Let $(\nu_0,\Phi_0)$ and $(\nu_1, \Phi_1)$ be two uniquely regular pairs and suppose that $\Phi_0$ and 
$\Phi_1$ are not physically equivalent. Then the following statements are mutually exclusive. 
\begin{enumerate}[label=(\roman*)]
\itemsep0em
\item $\nu_1\in\cS_+(\nu_0,\Phi_1)$, {\sl i.e.,} approach to equilibrium holds for $(\fA,\alpha_{\Phi_1},\nu_0)$.
\item There exists a uniquely regular path $\lb\nu,\Psi\rb$, with $\Psi\in C^1([0,1],\cBr)$, connecting $(\nu_0,\Phi_0)$ to $(\nu_1,\Phi_1)$ and satisfying the path adiabatic theorem. 
\end{enumerate}
\end{thm}

\begin{proof} Approach to equilibrium has been characterized in terms of a strict increase of specific entropy.  According to~\cite[Theorem~2.20]{Jaksic2023a},  if $\nu_1\in\cS_+(\nu_0,\Phi_1)$, then $s(\nu_0)<s(\nu_1)$ since $\Phi_0$ and $\Phi_1$ are not physically equivalent. By Theorem~\ref{prop-sun-1}, this fact is incompatible with the existence of the path $\lb\nu,\Psi\rb$. 
	
Conversely, Theorem~\ref{prop-sun-1} shows that the existence of the path $\lb\nu,\Psi\rb$ implies the equality $s(\nu_0)=s(\nu_1)$. By~\cite[Theorem~2.20]{Jaksic2023a},  if $\nu_1\in\cS_+(\nu_0,\Phi_1)$, then $s(\nu_0)=s(\nu_1)$ implies $\nu_0=\nu_1$ and that $\Phi_0$ and $\Phi_1$ are physically equivalent, which contradicts the assumption of the theorem.
\end{proof}

%========================================================================
\section{Proofs of Propositions~\ref{prop-balance} and~\ref{prop:alpha}}
\label{sec-ln}
%========================================================================

%====================================================
\subsection{Proof of Proposition~\ref{prop-balance}}
\label{sec-proof-of-prop:balance}
%====================================================

Our standing assumption is that the map $V\in C^1([0,1],\cO)$ takes its value in the self-adjoint elements of $\cO$, $\dot V$ denoting its derivative. Denote by $[0,1]\ni s,t\mapsto\alpha_T^{s\to t}$ the 2-parameter family of $\ast$-automorphisms of $\cO$ uniquely determined by the Cauchy problem
$$
\frac1T\partial_t\alpha_T^{s\to t}(A)=\alpha_T^{s\to t}\circ\delta_{V(t)}(A),\qquad\alpha_T^{s\to s}(A)=A\in\Dom(\delta),\qquad s,t\in[0,1].
$$
We start with: 

\begin{lem}\label{lem-cauchy}
For $s,t\in[0,1]$, one has 
\be
\alpha_T^{s\to t}=\gamma_T^{s\to t}\circ\alpha_{V(t)}^{(t-s)T},
\label{eq:sandyday}
\ee
where $\gamma_T^{s\to t}$ is the inner $\ast$-automorphism of\/ $\cO$ given by
$\gamma_T^{s\to t}(A)=\Gamma_T^{s\to t}A\Gamma_T^{s\to t\ast}$, and the 2-para\-meter family of unitaries $\Gamma_T$ is the solution of the Cauchy problem
\be
\frac{\i}T\partial_t\Gamma_T^{s\to t}=\Gamma_T^{s\to t}\int_s^t\alpha_{V(t)}^{(r-s)T}(\dot V(t))\d r,
\qquad\Gamma_T^{s\to s}=\one\qquad s,t\in[0,1].
\label{eq:GammaDef}
\ee
Moreover, one has $\Gamma_T^{s\to t}\in\Dom(\delta)$ for all $s,t\in[0,1]$.
\end{lem}
\begin{proof}
Fix $T>0$, $s\in[0,1]$, and set
$$
\widetilde{H}_t=T\int_s^t\alpha_{V(t)}^{(r-s)T}(\dot V(t))\d r,
$$
so that
\be
\partial_t\gamma_T^{s\to t}(A)=\gamma_T^{s\to t}(-\i[\widetilde{H}_t,A]).
\label{eq:D1}
\ee
For $A\in\Dom(\delta)$, one has $\partial_u\delta_{V(u)}(A)=\i[\dot V(u),A]$, and Duhamel's formula yields
$$
\partial_u\alpha_{V(u)}^{(t-s)T}(A)
=T\int_{s}^{t}\alpha_{V(u)}^{(r-s)T}(\i[\dot V(u),\alpha_{V(u)}^{(t-r)T}(A)])\d r,
$$
so that
$$
\partial_u\alpha_{V(u)}^{(t-s)T}(A)\bigg|_{u=t}=\i[\widetilde{H}_t,\alpha_{V(t)}^{(t-s)T}(A)],
$$
and hence
\be
\partial_t\alpha_{V(t)}^{(t-s)T}(A)=T\alpha_{V(t)}^{(t-s)T}\circ\delta_{V(t)}(A)
+\i[\widetilde{H}_t,\alpha_{V(t)}^{(t-s)T}(A)].
\label{eq:D2}
\ee
Combining Relations~\eqref{eq:D1} and~\eqref{eq:D2}, we derive
\begin{align*}
\partial_t\gamma_T^{s\to t}\circ\alpha_{V(t)}^{(t-s)T}(A)
&=\gamma_T^{s\to t}(-\i[\widetilde{H}_t,\alpha_{V(t)}^{(t-s)T}(A)])
+\gamma_T^{s\to t}(T\alpha_{V(t)}^{(t-s)T}\circ\delta_{V(t)}(A)
+\i[\widetilde{H}_t,\alpha_{V(t)}^{(t-s)T}(A)])\\
&=T\gamma_T^{s\to t}\circ\alpha_{V(t)}^{(t-s)T}\circ\delta_{V(t)}(A).
\end{align*}
Since obviously  $\gamma_T^{s\to s}\circ\alpha_{V(t)}^{(s-s)T}(A)=A$, Relation~\eqref{eq:sandyday} holds.

To prove the last statement, for fixed $s,t\in[0,1]$ let the sequence $(Q_n)_{n\in\nn}\subset\Dom(\delta_{V(t)})$ be such that $\lim_{n\to\infty}Q_n=\dot V(t)$. Setting
$$
H_n=T\int_s^t\alpha_{V(t)}^{(r-s)T}(Q_n)\d r,
$$
and observing that
$$
T\int_s^t\alpha_{V(t)}^{(r-s)T}(\delta_{V(t)}(Q_n))\d r=\int_s^t\partial_r\alpha_{V(t)}^{(r-s)T}(Q_n)\d r=\alpha_{V(t)}^{(t-s)T}(Q_n)-Q_n,
$$
we derive
$$
\delta_{V(t)}(H_n)=\alpha_{V(t)}^{(t-s)T}(Q_n)-Q_n.
$$
From the obvious facts that
$$
\lim_{n\to\infty}H_n=\widetilde{H}_t,\qquad
\lim_{n\to\infty}\delta_{V(t)}(H_n)=\alpha_{V(t)}^{(t-s)T}(\dot V(t))-\dot V(t),
$$
the closedness of $\delta_{V(t)}$ allows us to conclude that $\widetilde{H}_t\in\Dom(\delta_{V(t)})=\Dom(\delta)$ and that
\be
\delta_{V(t)}(\widetilde{H}_t)=\alpha_{V(t)}^{(t-s)T}(\dot V(t))-\dot V(t).
\label{eq:delta(H)}
\ee
Writing the solution to the Cauchy problem~\eqref{eq:GammaDef} as the uniformly convergent Dyson series
$$
\Gamma_T^{s\to t}=\sum_{n\ge0}(-\i)^n\int\limits_{s\le t_1\le\cdots\le t_n\le t}\widetilde{H}_{t_1}\cdots\widetilde{H}_{t_n}\d t_1\cdots\d t_n,
$$
one deduces $\Gamma_T^{s\to t}\in\Dom(\delta)$ for all $s,t\in[0,1]$ and that 
$$
\delta(\Gamma_T^{s\to t})
=\sum_{n\ge0}(-\i)^n\sum_{k=1}^n\ \int\limits_{s\le t_1\le\cdots\le t_n\le t}\widetilde{H}_{t_1}\cdots\delta(\widetilde{H}_{t_k})
\cdots\widetilde{H}_{t_n}\d t_1\cdots\d t_n,
$$
where the series on right hand side  converges uniformly.
\end{proof}
We now proceed with the proof of Proposition~\ref{prop-balance}.

The previous lemma gives
$$
S(\omega_{V(s)}\circ\alpha_T^{s\to t}|\omega_{V(t)})
=S(\omega_{V(s)}\circ\gamma_T^{s\to t}\circ\alpha_{V(t)}^{(t-s)T}|\omega_{V(t)})
=S(\omega_{V(s)}\circ\gamma_T^{s\to t}|\omega_{V(t)}),
$$
and invoking~\cite[Theorem~1.1]{Jaksic2003}\footnote{We note that the definition of relative entropy used in~\cite{Jaksic2003,Derezinski2003a} differs in its sign with the one used here.}, we can write
\be
S(\omega_{V(s)}\circ\alpha_T^{s\to t}|\omega_{V(t)})
=S(\omega_{V(s)}|\omega_{V(t)})
-\i\beta\omega_{V(s)}(\Gamma_T^{s\to t}\delta_{V(t)}(\Gamma_T^{s\to t\ast})).
\label{eq:Ssplit}
\ee
To deal with the first term on the right-hand side, we invoke Araki's perturbation theory (see, e.g., \cite[Theorem~5.1]{Derezinski2003a}).
Setting  $\nu=\omega_{V(s)}$ and $W_t=V(t)-V(s)$, we have
$$
S(\omega_{V(s)}|\omega_{V(t)})=S(\nu|\nu_{W_t})=\beta\nu(W_t)+\log\langle\Omega_\nu,\e^{-\beta(L+\pi_\nu(W_t))}\Omega_\nu\rangle,
$$
where $(\cH_\nu,\pi_\nu,\Omega_\nu)$ denotes the GNS representation of $\cO$ induced by $\nu$, and $L$ is the standard Liouvillean generating the unitary implementation of the $C^\ast$-dynamics $\alpha_{V(s)}$ on $\cH_\nu$. Taking the derivative w.r.t.\;$t$ of the logarithmic term and using the facts that $W_s=0$ and $L\Omega_\nu=0$ gives
$$
\log\langle\Omega_\nu,\e^{-\beta(L+\pi_\nu(W_t))}\Omega_\nu\rangle
=\int_s^t\partial_u\log\langle\Omega_\nu,\e^{-\beta(L+\pi_\nu(W_u))}\Omega_\nu\rangle\d u.
$$
Duhamel's formula further yields
\begin{align*}
\partial_u\log\langle\Omega_\nu,\e^{-\beta(L+\pi_\nu(W_u))}\Omega_\nu\rangle
&=-\int_0^\beta\frac{\langle\Omega_\nu,\e^{-\gamma(L+\pi_\nu(W_u))}
\pi_\nu(\dot V(u))\e^{-(\beta-\gamma)(L+\pi_\nu(W_u))}\Omega_\nu\rangle}{\|\Psi_u\|^2}\d\gamma\\[6pt]
&=-\int_{-\beta/2}^{\beta/2}\frac{\langle\Psi_u,\e^{-\gamma(L+\pi_\nu(W_u))}
\pi_\nu(\dot V(u))\e^{\gamma(L+\pi_\nu(W_u))}\Psi_u\rangle}{\|\Psi_u\|^2}\d\gamma,
\end{align*}
where $\Psi_u=\e^{-\beta(L+\pi_\nu(W_u))/2}\Omega_\nu$ is, up to normalization, the vector representative of $\omega_{V(u)}$ in $\cH_\nu$. It follows that
$$
\partial_u\log\langle\Omega_\nu,\e^{-\beta(L+\pi_\nu(W_u))}\Omega_\nu\rangle
=-\int_{-\beta/2}^{\beta/2}\omega_{V(u)}\left(\alpha_{V(u)}^{\i\gamma}(\dot V(u))\right)\d\gamma
=-\beta\omega_{V(u)}(\dot V(u)),
$$
and hence
$$
S(\omega_{V(s)}|\omega_{V(t)})=\beta\omega_{V(s)}(V(t)-V(s))-\beta\int_s^t\omega_{V(u)}(\dot V(u))\d u
=-\beta\int_s^t(\omega_{V(u)}-\omega_{V(s)})(\dot V(u))\d u.
$$
To deal with the second term on the right-hand side of~\eqref{eq:Ssplit}, we note that
$$
\i\Gamma_T^{s\to t}\delta_{V(t)}(\Gamma_T^{s\to t\ast})
=\int_s^t\i\partial_u\Gamma_T^{s\to u}\delta_{V(u)}(\Gamma_T^{s\to u\ast})\d u,
$$
and so, using Relation~\eqref{eq:delta(H)}, we have that 
\begin{align*}
\i\partial_u\Gamma_T^{s\to u}\delta_{V(u)}(\Gamma_T^{s\to u\ast})
&=\Gamma_T^{s\to u}\widetilde{H}_u\delta_{V(u)}(\Gamma_T^{s\to u\ast})
-\Gamma_T^{s\to u}\delta_{V(u)}(\widetilde{H}_u\Gamma_T^{s\to u\ast})
-\Gamma_T^{s\to u}[\dot V(u),\Gamma_T^{s\to u\ast}]\\[4pt]
&=-\Gamma_T^{s\to u}(\delta_{V(u)}(\widetilde{H}_u)+\dot V(u))\Gamma_T^{s\to u\ast}+\dot V(u)\\[4pt]
&=-\Gamma_T^{s\to u}\alpha_{V(u)}^{(u-s)T}(\dot V(u))\Gamma_T^{s\to u\ast}+\dot V(u)\\[4pt]
&=-\alpha_T^{s\to u}(\dot V(u))+\dot V(u).
\end{align*}
It follows that
$$
-\i\beta\omega_{V(s)}(\Gamma_T^{s\to t}\delta_{V(t)}(\Gamma_T^{s\to t\ast}))
=\beta\int_s^t\left(\omega_{V(s)}\circ\alpha_T^{s\to u}(\dot V(u))-\omega_{V(s)}(\dot V(u))\right)\d u,
$$
and finally, Relation~\eqref{eq:Ssplit} gives
$$
S(\omega_{V(s)}\circ\alpha_T^{s\to t}|\omega_{V(t)})=
\beta\int_s^t
\left(\omega_{V(s)}\circ\alpha_T^{s\to u}(\dot V(u))-\omega_{V(u)}(\dot V(u))\right)\d u.
$$

%==================================================
\subsection{Proof of Proposition~\ref{prop:alpha}}
\label{sec-proof-of-prop:alpha}
%==================================================

We start with some preliminary observations. By~\cite[Theorem~6.2.4]{Bratteli1981}, the time-dependent derivation $[0,1]\ni\tau\mapsto\delta_{\Phi_\tau}$ defined on $\fAloc$ by
$$
\delta_{\Phi_\tau}(A)=\sum_{X\cap\supp(A)\neq\emptyset}\i[\Phi_\tau(X),A],
$$
is closable and its closure, which we shall denote by the same symbol, generates the frozen $C^\ast$-dynamics
$$
\rr\ni t\mapsto\alpha_{\Phi_\tau}^t=\e^{t\delta_{\Phi_\tau}}.
$$
The translation invariance of the interaction $\Phi_\tau$ implies that $\varphi^x(\Dom(\delta_{\Phi_\tau}))=\Dom(\delta_{\Phi_\tau})$ and that
\be
\delta_{\Phi_\tau}\circ\varphi^x=\varphi^x\circ\delta_{\Phi_\tau},
\label{equ:deltatrans}
\ee
for all $x\in\zz^d$.

Note that, for $A\in\fAloc$, one has
$$
\delta_{\Phi_\tau}(A)=\lim_{\Lambda\uparrow\zz^d}\delta_{\Phi_{\tau},\Lambda}(A),
$$
 where the convergence is in norm and 
$$
\delta_{\Phi_{\tau},\Lambda}(A)=\sum_{X\subset\Lambda}\i[\Phi_\tau(X),A]=\i[H_\Lambda(\Phi_\tau),A],
$$
is such that $\supp(\delta_{\Phi_{\tau},\Lambda}(A))\subset\supp(A)\cup\Lambda$, {\sl i.e.,} $\delta_{\Phi_{\tau},\Lambda}(\fAloc)\subset\fAloc$. It follows that the non-auto\-no\-mous $C^\ast$-dynamics generated by the time-dependent local Hamiltonian $H_\Lambda(\Phi_\tau)$ is given by the norm-convergent Dyson expansion
\be
\alpha_{\Phi,\Lambda}^{\sigma\to\tau}(A)=\begin{cases}
\displaystyle A+\sum_{n\ge1}\int_{\Delta_n(\sigma,\tau)}
\delta_{\Phi_{\tau_1},\Lambda}\circ\cdots\circ\delta_{\Phi_{\tau_n},\Lambda}(A)
\,\d\tau_1\cdots\d\tau_n,&\text{ for }0\le\sigma\le\tau\le1;\\[12pt]
\displaystyle A+\sum_{n\ge1}\int_{\Delta_n(\sigma,\tau)}
\delta_{\Phi_{\tau_n},\Lambda}\circ\cdots\circ\delta_{\Phi_{\tau_1},\Lambda}(A)
\,\d\tau_1\cdots\d\tau_n,&\text{ for }0\le\tau\le\sigma\le1;
\end{cases}
\label{equ:DysonLambda}
\ee
for $A\in\fAloc$, where $\Delta_n(\sigma,\tau)=\{(\tau_1,\ldots,\tau_n)\mid0\le\min(\sigma,\tau)\le\tau_1\le\cdots\le\tau_n\le\max(\sigma,\tau)\}$.

In the following, the norm of a time-dependent continuous interaction $[0,1]\ni \tau \mapsto \Phi(\tau)\in \cBr$  is taken to be
\be
\|\Phi\|_r=\sup_{\tau\in[0,1]}\|\Phi_\tau\|_r.
\label{eq:Jeanne}
\ee

\medskip\noindent(1) A trivial extension of~\cite[Lemma~III.3.5]{Israel1979} gives that, for $A\in\fAloc$ and $\tau_1,\ldots,\tau_n\in[0,1]$,
$$
\sum_{X_1,\ldots,X_n\in\cF}
\left\|\i[\Phi_{\tau_1}(X_1),\i[\Phi_{\tau_2}(X_2),\cdots,\i[\Phi_{\tau_n}(X_n),A]\cdots]]\right\|
\le\|A\|\e^{r|\supp(A)|}\left(\frac{2\|\Phi\|_r}{r}\right)^n n!.
$$
Hence,
\be
\sup_{\tau_1,\ldots,\tau_n\in[0,1]}\|\delta_{\Phi_{\tau_1}}\circ\cdots\circ\delta_{\Phi_{\tau_n}}(A)\|
\le\|A\|\e^{r|\supp(A)|}\left(\frac{2\|\Phi\|_r}{r}\right)^n n!,
\label{equ:deltas}
\ee
and the same estimate holds with  $\delta_{\Phi_{\tau_k}}$ replaced by $\delta_{\Phi_{\tau_k},\Lambda}$.
Thus, taking the thermodynamic limit in~\eqref{equ:DysonLambda} yields that, for $|\tau-\sigma|<\epsilon=r/2\|\Phi\|_r$,
\be
\lim_{\Lambda\uparrow\zz^d}\alpha_{\Phi,\Lambda}^{\sigma\to\tau}(A)
=\begin{cases}
\displaystyle A+\sum_{n\ge1}\int_{\Delta_n(\sigma,\tau)}
\delta_{\Phi_{\tau_1}}\circ\cdots\circ\delta_{\Phi_{\tau_n}}(A)
\,\d\tau_1\cdots\d\tau_n,&\text{ for }0\le\sigma\le\tau\le1;\\[12pt]
\displaystyle A+\sum_{n\ge1}\int_{\Delta_n(\sigma,\tau)}
\delta_{\Phi_{\tau_n}}\circ\cdots\circ\delta_{\Phi_{\tau_1}}(A)
\,\d\tau_1\cdots\d\tau_n,&\text{ for }0\le\tau\le\sigma\le1;
\end{cases}
\label{equ:Dyson}
\ee
the limit being uniform for $\tau-\sigma$ in compact subsets of ${]}-\epsilon,\epsilon{[}$. Since the maps $\alpha_{\Phi,\Lambda}^{\sigma\to\tau}:\fAloc\to\fA$ are isometric, their limit $\alpha_{\Phi}^{\sigma\to\tau}$ is norm continuous, and hence uniquely extends by continuity to an isometry on $\fA$. It follows that~\eqref{equ:alphaTDlimit} holds for all $A\in\fA$ and $\sigma,\tau\in[0,1]$ satisfying  $|\tau-\sigma|<\epsilon$. Moreover, as norm-limits of $\ast$-morphisms, the maps $\alpha_{\Phi}^{\sigma\to\tau}$ are themselves $\ast$-morphisms. For $\sigma,\tau,\upsilon\in[0,1]$ such that $\max(|\tau-\sigma|,|\sigma-\upsilon|)<\epsilon$, it follows from the continuity of $\alpha_{\Phi}^{\sigma\to\tau}$ that
$$
\lim_{\Lambda\uparrow\zz^d}\alpha_{\Phi}^{\sigma\to\tau}\circ\alpha_{\Phi,\Lambda}^{\tau\to\upsilon}(A)
=\alpha_{\Phi}^{\sigma\to\tau}\circ\alpha_{\Phi}^{\tau\to\upsilon}(A),
$$
for $A\in\fAloc$. Writing 
$$
\alpha_{\Phi}^{\sigma\to\tau}\circ\alpha_{\Phi,\Lambda}^{\tau\to\upsilon}(A)
=\alpha_{\Phi,\Lambda}^{\sigma\to\upsilon}(A)
+(\alpha_{\Phi}^{\sigma\to\tau}-\alpha_{\Phi,\Lambda}^{\sigma\to\tau})\circ\alpha_{\Phi}^{\tau\to\upsilon}(A)
+(\alpha_{\Phi}^{\sigma\to\tau}-\alpha_{\Phi,\Lambda}^{\sigma\to\tau})\circ
(\alpha_{\Phi,\Lambda}^{\tau\to\upsilon}-\alpha_{\Phi}^{\tau\to\upsilon})(A),
$$
and using that
\begin{gather*}
\lim_{\Lambda\uparrow\zz^d}(\alpha_{\Phi}^{\sigma\to\tau}-\alpha_{\Phi,\Lambda}^{\sigma\to\tau})\circ\alpha_{\Phi}^{\tau\to\upsilon}(A)=0,\\[4pt]
\lim_{\Lambda\uparrow\zz^d}\|(\alpha_{\Phi}^{\sigma\to\tau}-\alpha_{\Phi,\Lambda}^{\sigma\to\tau})\circ
(\alpha_{\Phi,\Lambda}^{\tau\to\upsilon}-\alpha_{\Phi}^{\tau\to\upsilon})(A)\|
\le2\lim_{\Lambda\uparrow\zz^d}\|\alpha_{\Phi,\Lambda}^{\tau\to\upsilon}(A)-\alpha_{\Phi}^{\tau\to\upsilon}(A)\|=0,
\end{gather*}
we derive
\be
\alpha_{\Phi}^{\sigma\to\tau}\circ\alpha_{\Phi}^{\tau\to\upsilon}(A)=\alpha_{\Phi}^{\sigma\to\upsilon}(A),
\label{equ:chain}
\ee
provided that $|\upsilon-\sigma|<\epsilon$. This identity extends by continuity to all $A\in\fA$, and allows  to extend the family of $\ast$-morphisms $(\alpha_\Phi^{\sigma\to\tau})_{|\tau-\sigma|<\epsilon}$ to arbitrary $\sigma,\tau\in[0,1]$ by setting
\be
\alpha_\Phi^{\sigma\to\tau}=\alpha_\Phi^{\sigma\to\tau_1}\circ\alpha_\Phi^{\tau_1\to\tau_2}\circ\cdots
\circ\alpha_\Phi^{\tau_{n-1}\to\tau_n}\circ\alpha_\Phi^{\tau_n\to\tau},
\label{equ:prodexp}
\ee
where $\tau_0=\sigma,\tau_1,\ldots,\tau_{n+1}=\tau$ are such that $\max_k|\tau_{k+1}-\tau_k|<\epsilon$. One easily checks that the left-hand side of this expression does not depend on the choice of the subdivisions $\tau_1,\ldots,\tau_n$ on the right-hand side. By construction, this extension satisfies~\eqref{equ:chain} for all $\sigma,\tau,\upsilon\in[0,1]$. From the telescopic expansion based on~\eqref{equ:prodexp},
$$
\alpha_\Phi^{\sigma\to\tau}(A)-\alpha_{\Phi,\Lambda}^{\sigma\to\tau}(A)
=\sum_{k=0}^n\alpha_{\Phi,\Lambda}^{\sigma\to\tau_1}\circ\cdots\circ\alpha_{\Phi,\Lambda}^{\tau_{k-1}\to\tau_k}\circ
(\alpha_{\Phi}^{\tau_k\to\tau_{k+1}}-\alpha_{\Phi,\Lambda}^{\tau_k\to\tau_{k+1}})\circ
\alpha_{\Phi}^{\tau_{k+1}\to\tau_{k+2}}\circ\cdots\circ\alpha_{\Phi}^{\tau_n\to\tau}(A),
$$
we derive
$$
\|\alpha_\Phi^{\sigma\to\tau}(A)-\alpha_{\Phi,\Lambda}^{\sigma\to\tau}(A)\|\le
\sum_{k=0}^n\|(\alpha_{\Phi}^{\tau_k\to\tau_{k+1}}-\alpha_{\Phi,\Lambda}^{\tau_k\to\tau_{k+1}})(A_k)\|,
$$
with $A_k=\alpha_{\Phi}^{\tau_{k+1}\to\tau_{k+2}}\circ\cdots\circ\alpha_{\Phi}^{\tau_n\to\tau}(A)$. It follows that~\eqref{equ:alphaTDlimit} holds for all $A\in\fA$ and all $\sigma,\tau\in[0,1]$. Note that~\eqref{equ:chain} implies that $\alpha_\Phi^{\sigma\to\tau}\circ\alpha_\Phi^{\tau\to\sigma}=\alpha_\Phi^{\sigma\to\sigma}=\mathrm{Id}$, which shows that the maps $\alpha_\Phi^{\sigma\to\tau}$ are $\ast$-automorphisms of $\fA$. 

\medskip\noindent(2) The translation invariance $\alpha_{\Phi}^{\sigma\to\tau}\circ\varphi^x=\varphi^x\circ\alpha_{\Phi}^{\sigma\to\tau}$ follows immediately from Relation~\eqref{equ:deltatrans}.

\medskip\noindent(3) For $A\in\fAloc$ and $|\tau-\sigma|<\epsilon/2$, combining the Dyson expansion~\eqref{equ:Dyson} with the estimate~\eqref{equ:deltas}, we get
$$
\|\alpha_\Phi^{\sigma\to\tau}(A)-A\|\le2\epsilon^{-1}\|A\|\e^{d|\supp(A)|}|\tau-\sigma|.
$$
Since $\alpha_\Phi^{\sigma\to\tau}$ is an isometry and $\fAloc$ is dense in $\fA$, we conclude that the map $(\sigma,\tau)\mapsto\alpha_\Phi^{\sigma\to\tau}$ is strongly continuous on the diagonal $(\tau,\tau)\in[0,1]\times[0,1]$. For $(\sigma_0,\tau_0)\in[0,1]\times[0,1]$, Property~\eqref{equ:chain} gives
$$
\alpha_\Phi^{\sigma\to\tau}(A)-\alpha_\Phi^{\sigma_0\to\tau_0}(A)
=\alpha_\Phi^{\sigma\to\tau_0}(\alpha_\Phi^{\tau_0\to\tau}(A)-A)+(\alpha_\Phi^{\sigma\to\sigma_0}(B)-B),
$$
with $B=\sigma_\Phi^{\sigma_0\to\tau_0}(A)$, which shows strong continuity at $(\sigma_0,\tau_0)$.

From~\eqref{equ:chain} and the Dyson expansion~\eqref{equ:Dyson}, we further derive that
$$
\partial_\tau\alpha_{\Phi}^{\sigma\to\tau}(A)=
\lim_{\upsilon\to0}\frac1\upsilon
\left(\alpha_{\Phi}^{\sigma\to\tau+\upsilon}(A)-\alpha_{\Phi}^{\sigma\to\tau}(A)\right)
=\lim_{\upsilon\to0}\frac1\upsilon\alpha_{\Phi}^{\sigma\to\tau}(\alpha_{\Phi}^{\tau\to\tau+\upsilon}(A)-A)
=\alpha_{\Phi}^{\sigma\to\tau}\circ\delta_{\Phi_\tau}(A),
$$
holds in norm for arbitrary $\sigma,\tau\in[0,1]$ and $A\in\fAloc$. 

\medskip\noindent(4) We will consider the case $\sigma<\tau$, the opposite case  is similar. We also observe that Part~(1) extends to time-dependent piecewise continuous interactions $\Phi$. Indeed, the continuity was only tacitly used to justify the Riemann integration in the Dyson expansions~\eqref{equ:DysonLambda} and~\eqref{equ:Dyson}. 

Let $\Phi,\Psi:[0,1]\to\cBr$ be piecewise continuous and set $\bar\epsilon=r/2\max(\|\Phi\|_r,\|\Psi\|_r)$. The telescopic expansion
$$
\delta_{\Phi_{\tau_1}}\circ\cdots\circ\delta_{\Phi_{\tau_n}}
-\delta_{\Psi_{\tau_1}}\circ\cdots\circ\delta_{\Psi_{\tau_n}}
=\sum_{k=1}^n\delta_{\Phi_{\tau_1}}\circ\cdots\circ\delta_{\Phi_{\tau_{k-1}}}
\circ\delta_{(\Phi-\Psi)_{\tau_k}}
\circ\delta_{\Psi_{\tau_{k+1}}}\circ\cdots\circ\delta_{\Psi_{\tau_n}},
$$
combined with the estimate~\eqref{equ:deltas}, gives
$$
\sup_{\tau_1,\ldots,\tau_n\in[0,1]}\|\delta_{\Phi_{\tau_1}}\circ\cdots\circ\delta_{\Phi_{\tau_n}}(A)
-\delta_{\Psi_{\tau_1}}\circ\cdots\circ\delta_{\Psi_{\tau_n}}(A)\|
\le2n\|A\|\e^{r|\supp(A)|}\bar\epsilon^{-n+1}\frac{\|\Phi-\Psi\|_r}{r}n!,
$$
for $A\in\fAloc$. Thus, it follows from the Dyson expansion~\eqref{equ:Dyson} that
$$
\sup_{|\tau-\sigma|\le\bar\epsilon/2}\|\alpha_\Phi^{\sigma\to\tau}(A)-\alpha_\Psi^{\sigma\to\tau}(A)\|
\le\frac{4\bar\epsilon}r\|A\|\e^{r|\supp(A)|}\|\Phi-\Psi\|_r.
$$
By density/continuity, we derive that
\be
\lim_{\Psi\to\Phi}\|\alpha_\Phi^{\sigma\to\tau}(A)-\alpha_\Psi^{\sigma\to\tau}(A)\|=0
\label{equ:PhiCont}
\ee
for all $A\in\fA$ and $|\tau-\sigma|\le\bar\epsilon/2$. Finally, a telescopic expansion analogous to~\eqref{equ:prodexp} and with $\max_k(|\tau_{k+1}-\tau_k|)<\bar\epsilon/2$ gives
$$
\|\alpha_\Phi^{\sigma\to\tau}(A)-\alpha_{\Psi}^{\sigma\to\tau}(A)\|\le
\sum_{k=0}^n\|(\alpha_{\Phi}^{\tau_k\to\tau_{k+1}}-\alpha_{\Psi}^{\tau_k\to\tau_{k+1}})(A_k)\|,
$$
with $A_k=\alpha_{\Phi}^{\tau_{k+1}\to\tau_{k+2}}\circ\cdots\circ\alpha_{\Phi}^{\tau_n\to\tau}(A)$, from which we conclude that~\eqref{equ:PhiCont} holds for all $A\in\fA$ and all $\sigma,\tau\in[0,1]$.

To complete the proof, given the continuous time-dependent interaction $\Phi$, set
$$
\Phi_t^{(N)}=1_{[0,\upsilon_0{[}}(t)\Phi(t)+
\sum_{k=0}^{N-1}1_{[\upsilon_k,\upsilon_{k+1}{[}}(t)\Phi_{\upsilon_k}
+1_{[\upsilon_N,1]}(t)\Phi_t,
$$
where $1_I$ denotes the indicator function of the set $I$. Since
$$
\sup_{t\in[0,1]}\|\Phi_t-\Phi^{(N)}_t\|_r\le\sup_{|t-s|<\xi_N}\|\Phi_t-\Phi_s\|_r,
$$
the uniform continuity of the map $[0,1]\ni\tau\mapsto\Phi_\tau\in\cBr$ implies
$$
\lim_{N\to\infty}\|\Phi-\Phi^{(N)}\|_r=0.
$$
Thus, the claim follows from~\eqref{equ:PhiCont} and the fact that
$$
\alpha_{\Phi^{(N)}}^{\sigma\to\tau}(A)=\alpha_{\Phi_{\upsilon_0}}^{\xi_N}\circ\cdots\circ\alpha_{\Phi_{\upsilon_{N-1}}}^{\xi_N}(A).
$$

%===========================
\appendix
\section{Glossary of terms}
\label{app:notations}
%===========================

More details can be found in~\cite{Jaksic2023a}.

\begin{itemize}
\item $\fA$: the $C^\ast$-algebra of a spin system on $\zz^d$, or equivalently, the gauge-invariant sector of the fermionic CAR algebra on $\ell^2(\zz^d)$.
\item $\cF$: the set of all finite subsets of $\mathbb Z^d$.
\item $\Phi$: an interaction, namely a translation invariant family $\{\Phi(X)\}_{X\in\cF}$ of self-adjoint elements of $\fA$ with $\supp(\Phi)=X$.
\item $\cBr$: the Banach space of interactions satisfying $\displaystyle\|\Phi\|_r=\sum_{X\ni0}\e^{r(|X|-1)}\|\Phi(X)\|<\infty$, ($r>0$).
\item $\zz^d\ni x\mapsto\varphi^x$: the group action of $\zz^d$ on $\fA$. 
\item $\cSI(\fA)$: the set of translation invariant states on $\fA$.
\item $s(\nu)$: the specific entropy of a state $\nu\in\cSI(\fA)$.
\item $s(\nu|\omega)$: the specific relative entropy of two states $\nu,\omega\in\cSI(\fA)$.
\item $P(\Phi)$: the pressure of the interaction $\Phi$.
\item $E_\Phi= \sum_{X\ni0} |X|^{-1} \Phi(X)$, so that $\nu(E_\Phi)$ is the expected specific energy of the interaction $\Phi$ in the state $\nu\in\cSI(\fA)$.
\item $\cSeq(\Phi)=\{\nu\in\cSI(\fA)\mid P(\Phi)=s(\nu)-\nu(E_\Phi)\}$: the set of equilibrium states for $\Phi$ (from the Gibbs variational principle).
\item $\cSwg(\Phi)$: the set of weak Gibbs states for $\Phi$, see~\cite[Section~2.2]{Jaksic2023a}.
\item $\alpha_\Phi$: the Heisenberg dynamics generated by the (time-independent) interaction $\Phi$.
\item  Physical equivalence: two interactions $\Phi, \Psi\in \cBr$ are physically equivalent  iff  $\alpha_\Phi=\alpha_\Psi$, see~\cite[Theorem~2.7]{Jaksic2023a}. 
\end{itemize}

\bibliographystyle{capalpha}
\bibliography{MASTER}

\end{document}